\documentclass[]{article}

\usepackage[letterpaper,margin=2cm]{geometry} %page layout
\usepackage[latin1]{inputenc}
\usepackage[english]{babel}
\usepackage{amsmath} %it's amsmath, no justification is necessary
\usepackage{amsthm} %particularly for proof environment
\usepackage{amsfonts} %for \mathbb
\usepackage{amssymb} %for\mathfrak
\usepackage{graphicx} %for \includegraphics
\usepackage{color}
\usepackage{epstopdf} %for automatically including .eps with \includegraphics
	\parskip\medskipamount
\newtheorem{theo}{Theorem} %[section]
\newtheorem{cor}[theo]{Corollary}

\theoremstyle{definition}
\newtheorem{defn}[theo]{Definition}
\newcommand{\R}{\mathbb{R}}

\newcommand{\A}{\mathfrak{L}}

\newcommand\sS{{\mathcal S}}
\newcommand\given{\;|\;}
\newcommand\suchthat{\;:\;} % such that

\newcommand\tildem{{\tilde m}}
\newcommand\tildes{{\tilde s}}

\newcommand\norm[1]{\left\Vert#1\right\Vert}
\let\ds\displaystyle
\newcommand\proj{\mathop{\rm proj}}

\begin{document}

%Begin front matter
\title{A Markovian genomic concatenation model guided by persymmetric matrices}

\author{Andrew G. Hart$^{\dag}$
 \and Marcelo Sobottka$^\ddag$\footnote{\noindent Corresponding author.\newline
E-mail addresses: \texttt{ahart@dim.uchile.cl} (A. Hart),
\texttt{marcelo.sobottka@ufsc.br} (M. Sobottka).}\\
{\footnotesize $^\dag$ Departamento de Ingenier\'{\i}a Matem\'{a}tica and Centro de
Modelamiento Matem\'{a}tico, Universidad de Chile, Chile.}\\
{\footnotesize $^\ddag$ Departamento de Matem\'{a}tica, Universidade Federal de
Santa Catarina, Brazil} }

\date{\ }
\maketitle

\medskip
\hrule
\medskip

\begin{abstract}
The aim of this work is to provide a rigorous mathematical analysis of a
stochastic concatenation model presented by Sobottka and Hart (2011) which allows approximation of the
first-order stochastic structure in bacterial DNA by means of a stationary
Markov chain. Two probabilistic constructions that
rigorously formalize the model are presented. Necessary and sufficient
conditions for a Markov chain to be generated by the model are given, as well as the theoretical background
needed for designing new algorithms for statistical analyses of real
bacterial genomes. It is shown that the model encompasses the Markov chains
satisfying intra-strand parity, a property observed in most DNA
sequences.
\end{abstract}

{\small\noindent

}

\section{Introduction}
\label{sec:intro}

One of the fundamental issues in the study of genomes is their primary
structure, that is, the distribution of nucleotides along DNA sequences.  The identification of
statistical patterns in the primary structure of DNA sequences has revealed several underlying patterns in genomes \cite{Cattani12,Li_92,Lobry96-2,sobottka&hart2011} and has enabled scientists to propose models for
evolutive pressures and mutational mechanisms that might act on organisms
\cite{AlbrechtBuehler2006,hart&martinez&olmos2012,sobottka&hart2011} as well as
to construct  bioinformatics tools.
For example, in \cite{Felsenstein81}, a maximum likelihood approach was used to
perform analyses of DNA sequences in order to estimate evolutionary trees,
while in \cite{Yu_et_al2000}, a measure of the long-range correlation between
the nucleotide bases of DNA sequences was used to classify bacteria. In
addition, strand compositional asymmetry (SCA) was used to detect replication
origins in bacteria \cite{FrankLobry00}, while \cite{Salzberg_et_al98} used
interpolated Markov models to identify genes in bacteria,
\cite{hart&martinez&videla2006} proposed a maximization model to describe the
organization and distribution of genes in bacterial DNA and \cite{martinez2016}
presented a stationary stochastic process for modeling the placement of coding
and non-coding regions within a genome that incorporates the phenomenon of
start codons appearing within coding regions.

The aim of this work is to provide a rigorous formalization of a
stochastic concatenation model for capturing the primary structure of bacterial DNA sequences which
was presented in \cite{sobottka&hart2011}. The model, henceforth referred to as
the S-H model, allowed novel statistical symmetries in the mononucleotide and dinucleotide distributions of a collection of bacterial
chromosomes to be observed.
A key feature of the model is a persymmetric
matrix of probabilities which plays a role in determining the nucleic acids seen
along a DNA sequence. The persymmetric matrices constitute a special
class of matrices which has been employed in models from various fields (see for
example \cite{Nian1997, Nian&Chu1994,Nield1994}) and which has been widely
studied (see for example \cite{Gutierrez2014,Huang&Cline1972,Reid1997,Xie&Sheng2003}).

A genome is a duplex of DNA strands, each strand consisting of a sequence of
nucleotides. The nucleotides are of four types: adenine ($A$), cytosine ($C$),
guanine ($G$) and thymine ($T$). Of these types, adenine is complementary to
thymine while cytosine is complementary to guanine. Each
nucleotide on one DNA strand pairs with its complement on the opposite strand. This
chemically induced pairing between the two strands causes the strands to assume a ladder-like
arrangement which is then twisted to attain the famous helix.
The chemical composition of DNA molecules endows a strand with an intrinsic
reading direction: each strand can only be read in one direction by the genetic machinery of the cell. Furthermore, the way strands combine to form a
duplex means that the two strands  are read in opposite
directions: they are said to be antiparallel.

We shall identify each nucleotide type with a
number in $N:=\{1,2,3,4\}$ ($A\equiv 1$, $C\equiv 2$, $G\equiv 3$ and $T\equiv
4$). Let
$\alpha:N\rightarrow N$ be the involution which maps each nucleotide to its
complement, that is, $\alpha(i)=5-i$.

The S-H model is a concatenation model which has at its core
a first-order Markov chain whose
one-step transition matrix $P=\bigl(P_{ij}\bigr)_{i,j\in N}$ is derived from a
positive parameter $m$ and a positive persymmetric matrix $\A=\bigl(L_{ij}\bigr)_{i,j\in
N}$:
\begin{equation}
\label{p.form}
P_{ij}=\frac{L_{ij}M_j}{\sum_{k\in N} L_{ik}M_k},
\end{equation}
where $M_1=M_4:=m/(2m+2)$ and $M_2=M_3:=1/(2m+2)$.
The Markov chain governs how the DNA sequence grows in both directions from
an initial nucleotide called the origin
by appending nucleotides
in three steps. {\bf Step 1.}\ a nucleotide of type $j$ is randomly selected
with probability $M_j$. {\bf Step 2.}\ with probability $1/2$, the nucleotide
tries to join the end (consonant with the DNA reading direction ) or
beginning (contrary to the reading direction) of the sequence. {\bf Step 3.} In
the first case, the nucleotide is appended to the sequence with probability
$L_{ij}$, where $i$ is the type of the last nucleotide in the sequence; in
the latter, the nucleotide is prepended to the sequence with probability
$L_{\alpha(k)\alpha(j)}$, where $k$ is the type of the initial nucleotide.
This scheme is
illustrated in Figure \ref{fig:nucleotide_aggregation}.

Provided nucleotides accumulate evenly at the ends of the DNA strand, after
a long time one would obtain (with probability~$1$) a sequence with the initial
nucleotide at its midpoint. One half would be generated by the stationary Markov
chain $(P,\pi)$, where the transition matrix~$P$ is given by~\eqref{p.form} and
the chain's stationary distribution~$\pi$ is the left eigenvector of~$P$
corresponding to the eigenvalue~$1$ normalised to sum to~$1$. The other half
would have distribution given by the stationary Markov chain $(\tilde P,\tilde
\pi)$ where $\tilde \pi_i=\pi_{\alpha(i)}$ and $\tilde
P_{ij}=\frac{\pi_{\alpha(j)}}{\pi_{\alpha(i)}}P_{\alpha(j)\alpha(i)}$, for
$i,j\in N$.
The model is consistent with the observation reported by geneticists that
bacterial DNA sequences are usually composed of two distinct segments called chirochores (see \cite{FrankLobry00}). Furthermore, if one
estimates the transition matrices~$\tilde P$ and~$P$ for the segments prior to
and following the origin nucleotide respectively, one usually finds that
$\tilde
P_{ij}\approx\frac{\pi_{\alpha(j)}}{\pi_{\alpha(i)}}P_{\alpha(j)\alpha(i)}$ (see
\cite{sobottka&hart2011}).

\begin{figure}[!ht]
\centering
\includegraphics[width=.7\linewidth]{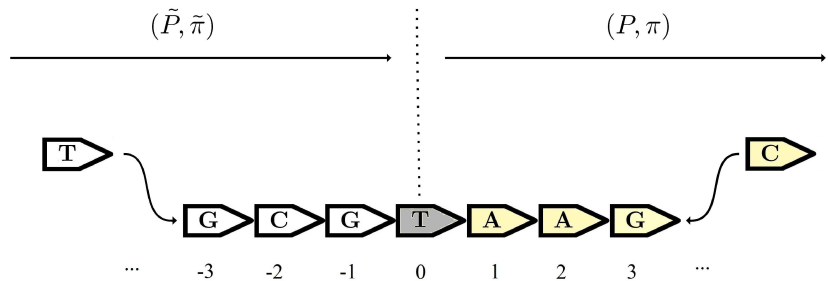}
\caption{A schematic presentation of the S-H model for constructing bacterial
DNA sequences. Assuming the reading sense of the sequence is from left to
right, a new nucleotide of type $C$ is selected with probability $1/(2m+2)$ and
is appended to the end of the sequence with probability $L_{32}$, while a
nucleotide of type $T$ is selected with probability $m/(2m+2)$ and will be
attached to the beginning of the sequence with probability
$L_{\alpha(3)\alpha(4)}$. The final DNA sequence obtained is the concatenation
of two Markovian processes: one starting at position zero and extending to the
right, whose estimated transition matrix is $P$; and the other terminating at zero, whose estimated transition matrix is $\tilde
P$.}\label{fig:nucleotide_aggregation}
\label{fig:1}
\end{figure}

The paper is organized as follows. Section~\ref{sec:interp} discusses the
probabilistic interpretation of the form~\eqref{p.form} of the matrix~$P$ in
greater depth than~\cite{sobottka&hart2011}. Two different probabilistic
constructions are presented, the first of which provides the justification for
the description of DNA sequence growth given above.
Section~\ref{sec:aleph.generated} introduces the set of $\aleph$-generated
matrices as matrices of the form \eqref{p.form}, where $\aleph$ is the set of
positive persymmetric matrices, and establishes several algebraic
characterizations of such matrices.
The non-uniqueness of the persymmetric matrix~$\A$ and positive parameter~$m$
that define an $\aleph$-generated matrix is then considered in
Section~\ref{sec:families}, where a couple of equivalence relations on~$\aleph$
are considered. This leads to an examination of various properties of
$\aleph$-generated matrices as used in the S-H model in
Section~\ref{sec:properties}.
Finally, we discuss some measures for determining how closely a DNA sequence
conforms to the S-H model and make concluding remarks in
Section~\ref{sec:conclusion}.

\section{Probabilistic interpretation of~$P$}
\label{sec:interp}

In~\cite{sobottka&hart2011}, a 	formal description of the way nucleotides are
appended to a DNA sequence using the persymmetric matrix $\A$ and the parameter
$m$ was presented, but the explicit connection with stochastic matrices of the
form~\eqref{p.form} was left for the reader to deduce. Here, we more rigorously
discuss how the form \eqref{p.form} of the stochastic matrix~$P$ arises from
the DNA-sequence growth mechanism described above. In addition, we shall present
an alternative probabilistic interpretation of the growth mechanism.

\subsection{Interpretation}

To begin, consider the growth of a DNA sequence whose initial nucleotide is
taken to be of type~$i$. Let $(\beta_t,\ t\geq0)$ be a Bernoulli scheme on~$N$
with common distribution $M=(M_1, M_2, M_3, M_4)=(m, 1, 1, m)\big/(2m+2)$, that
is, an independent and identically distributed sequence of random variables on $N$ with $\beta_s\sim M$. Consider two coupled stochastic processes
$(V_t,\ t\geq0)$, which evolves on the state space~$N$, and $(W_t,\ t\geq0)$,
which is a Bernoulli $\{0,1\}$-process where  $W_t$ is~$1$
with probability  $L_{V_t\beta_t}$ (that is,
$W_t \sim \textrm{B}\left(L_{V_t\beta_t}\right)$).
By setting $V_0:=i$ as the type of the initial nucleotide from which the DNA
sequence grows, the process $(V_t,\ t\geq0)$ evolves as a deterministic function of $(\beta_t,\ t\geq0)$ and $(W_t,\ t\geq0)$ as follows:
\begin{equation*}
V_{t+1} := \beta_tW_t + V_t(1-W_t)
= \begin{cases}
\beta_t, & \text{if } W_t=1 \\
V_t, & \text{if } W_t=0
\end{cases}\qquad,\qquad \forall t\geq 0.
\end{equation*}
Note that, while $V_t$ denotes the type of the last nucleotide appended to the
sequence, $\beta_t$ corresponds to the mechanism
responsible for proposing the type, say~$j$, of the next nucleotide to concatenate to the
sequence, and $W_t$ corresponds to the mechanism
responsible for accepting or rejecting the new nucleotide in the sequence. If $\beta_t=j$, then~$j$ is accepted as the type of the next
nucleotide provided that $W_t=1$, in which case $V_{t+1}$ is set to~$j$.
Otherwise, the nucleotide of type~$j$ is discarded and no nucleotide is
appended. In that case, $V_{t+1}$ takes the value of~$V_t$. In
this way,~$t$ counts the number of nucleotides proposed rather than the length
of the DNA sequence while the number of acceptances, given by $\sum_{u=1}^tW_u$,
is one less than the length of the DNA sequence, since it doesn't count the initial nucleotide.
For all $i\in N$ and $t\geq0$, we define
\begin{align*}
\gamma_i &:= \Pr(W_t=1 \given V_t=i)
= \sum_{j\in N} \Pr(W_t=1, \beta_t=j \given V_t=i) \\
&= \sum_{j\in N} \Pr(W_t=1\given \beta_t=j, V_t=i) \Pr(\beta_t=j \given V_t=i) \\
&= \sum_{j\in N} \Pr(W_t=1\given \beta_t=j, V_t=i) \Pr(\beta_t=j)
= \sum_{j\in N} L_{ij}M_j.
\end{align*}

Next, define a sequence $(\tau_s,\ s\geq0)$ of stopping times by $\tau_0:=0$ and
$$
\tau_{s+1} := \min\left\{t>\tau_s \suchthat W_{t-1}=1\right\}.
$$
The $\tau_s$'s mark the nucleotide type proposals that were accepted. By
construction, they constitute a series of renewal times. Note that $(V_t,\
t\geq0)$ is a discrete step function which transitions to a new nucleotide
whenever $t\in\{\tau_s,\ s\geq0\}$. More precisely, for all $s\geq0$,
$V_t=V_{\tau_s}$ for $t=\tau_s, \tau_s+1, \ldots, \tau_{s+1}-1$.
Let $i\in N$ and $w\in\{0,1\}$. The random variable $\beta_t$ is
independent of $W_u$ for $u<t$ and the distribution of $W_t$ is completely
determined by the value of $\beta_t$ and $V_t$. Consequently, the event
$\{W_t=w$ is conditionally independent of $\{W_u=0\}$ for all $u<t$ given
$V_t=i$.
For $i\in N$ and $t> u\geq0$, we have
\begin{align*}
\Pr(W_t=1, W_{t-1}=0, \ldots, W_u=0 \given V_u=i)
&= \Pr(W_t=1 \given W_{t-1}=0, \ldots, W_u=0,
V_u=i)
\cdot \\
&\qquad\ Pr(W_{t-1}=0, \ldots, W_u=0 \given V_u=i) \\
&= \Pr(W_t=1 \given V_t=i, W_{t-1}=0, \ldots,
W_u=0, V_u=0) \cdot \\
& \qquad
\Pr(W_{t-1}=0, \ldots, W_u=0 \given V_u=i) \\
&= \Pr(W_t=w \given V_t=i) \Pr(W_{t-1}=0, \ldots, W_u=0 \given V_u=i) \\
&= \gamma_i \Pr(W_{t-1}=0, \ldots, W_u=0 \given V_u=i)
\end{align*}
and
$$
\Pr(W_t=0, \ldots, W_u=0 \given V_u=i)
= (1-\gamma_i)\Pr(W_{t-1}=0, \ldots, W_u=0 \given V_u=i) .
$$
Hence, for $s\geq0$, $t\geq1$ and $i\in N$, we obtain
\begin{align*}
\Pr(\tau_{s+1}-\tau_s=t \given V_{\tau_s}=i)
&= \Pr(W_{\tau_s+t-1}=1, w_{\tau_s+t-2}=0, \ldots, W_{\tau_s}=0 \given
V_{\tau_s}=i) \\
&= \Pr(W_{\tau_s+t-2}=0, \ldots, W_{\tau_s}=0 \given V_{\tau_s}=i) \gamma_i \\
&= \Pr(W_{\tau_s+t-3}=0, \ldots, W_{\tau_s}=0 \given V_{\tau_s}=i)
(1-\gamma_i)\gamma_i
\\
&= \cdots \\
&= (1-\gamma_i)^{t-1}\gamma_i.
\end{align*}
Conditional on $V_{\tau_s}=i$, $\tau_{s+1}-\tau_s$ is thus a geometric random
variable taking values on the positive integers:
$$
\tau_{s+1}-\tau_s \given V_{\tau_s}=i \sim \textrm{geom}\bigl( \gamma_i\bigr),
\quad s\geq0,\ i\in N.
$$
Observe that the distribution of $\tau_{s+1}-\tau_s$ is completely determined by
the value of $V_{\tau_s}$ and is independent of any events prior to $\tau_s$ if
$V_{\tau_s}$ is given.
Furthermore, $\tau_{s+1}-\tau_s \given V_{\tau_s}=i$ is identically distributed
as $\tau_1 \given V_0=i$, for all $s>0$.

Next, define the process $(U_s,\ s\geq0)$ by $U_s := V_{\tau_s}$. Suppose that
$V_{\tau_s}=i$ for some fixed $s\geq0$. Then $V_{\tau_{s+1}}$ is determined by
$\beta_{\tau_{s+1}-1}$ and $V_{\tau_s}=\beta_{\tau_s-1}$, which are independent
of all $\beta_t$, $V_t$ and $W_t$ for all~$t$ prior to $\tau_s-1$. Consequently,
$(U_s,\ s\geq0)$ has the Markov property:
$$
\Pr(U_{s+1}=j \given U_s=i, U_{s-1}=i_1, \ldots, U_0=i_s) = \Pr(U_{s+1}=j
\given U_s=i),
$$
for all $i_1,i_2,\ldots, i_s\in N$ and $s\geq0$. Finally, since each $\tau_s$
essentially marks a point at which  the process $\bigl((\beta_t, V_t,
W_t),t\geq0\bigr)$ is restarted, we have
$$
\Pr(U_{s+1}=j \given U_s=i)
= \Pr(V_{\tau_{s+1}}=j \given V_{\tau_s}=i)
= \Pr(V_{\tau_1}=j \given V_{\tau_0}=i)
= \Pr(U_1=j \given U_0=i)=:P_{ij},
$$
for all $s\geq0$. Therefore, $(U_s,\ s\geq0)$ is a
time-homogeneous Markov chain on the finite state space~$N$. The following
theorem gives the form of the one-step transition matrix $P=\bigl( P_{ij} \bigr)_{i,j\in N}$ in terms
of~$\A$ and~$M$.

\begin{theo}
\label{thm:interp1}
The one-step transition matrix $P=\bigl(P_{ij}\bigr)_{i,j\in N}$ of the Markov
chain $(U_s,\ s\geq0)$ is given by
 $$
P_{ij}:=\frac{L_{ij}M_j}{\sum_{k\in N} L_{ik}M_k}.
$$
\end{theo}

\begin{proof}
Let $\tau:=\tau_1$. Now,
\begin{align}
\nonumber
P_{ij} &= \Pr( U_1=j \given U_0=i) \\
\nonumber
&= \Pr(V_\tau=j \given V_0=i) \\
\nonumber
&= \sum_{t=1}^\infty \Pr(V_t=j, \tau=t \given V_0=i) \\
\nonumber
&= \sum_{t=1}^\infty \frac{\Pr(V_t=j, \tau=t \given V_0=i)}{\Pr(\tau=t
\given V_0=i)}\Pr(\tau=t \given V_0=i) \\
\label{eqn:p.coin}
&= \sum_{t=1}^\infty \frac{\Pr(V_t=j, \tau=t \given V_0=i)}{\sum_{k\in N}
\Pr(V_t=k, \tau=t \given V_0=i)}\Pr(\tau=t \given V_0=i).
\end{align}
However,
\begin{align*}
\Pr(V_t=j, \tau=t & \given V_0=i) \\
&= \Pr(\beta_{t-1}=j, W_{t-1}=1, \tau=t \given V_0=i) \\
&= \Pr(\beta_{t-1}=j, W_{t-1}=1, W_u=0, u=1,\ldots,t-2 \given V_0=i) \\
&= \Pr(\beta_{t-1}=j, W_{t-1}=1 \given V_0=i, W_u=0, u=1,\ldots,t-2) \Pr(W_u=0,
u=1,\ldots,t-2 \given V_0=i) \\
&= \Pr(\beta_{t-1}=j, W_{t-1}=1 \given V_{t-1}=i) \Pr(W_u=0, u=1,\ldots,t-2
\given V_0=i) \\
&= \Pr(W_{t-1}=1 \given \beta_{t-1}=j, V_{t-1}=i) \Pr(\beta_{t-1}=j \given
V_{t-1}=i) \Pr(W_u=0, u=1,\ldots,t-2 \given V_0=i) \\
&= L_{ij}M_j \Pr(W_u=0, u=1,\ldots,t-2 \given V_0=i)
\end{align*}
and substituting this into \eqref{eqn:p.coin} yields
\begin{align*}
P_{ij}
&= \sum_{t=1}^\infty \frac{\Pr(V_t=j, \tau=t \given V_0=i)}{\sum_{k\in N} \Pr(V_t=k, \tau=t \given V_0=i)}\Pr(\tau=t \given V_0=i) \\
&= \sum_{t=1}^\infty \frac{L_{ij}M_j \Pr(W_u=0, u=1,\ldots,t-2 \given
V_0=i)}{\sum_{k\in N} L_{ik}M_k \Pr(W_u=0, u=1,\ldots,t-2 \given V_0=i)}
\Pr(\tau=t \given V_0=i) \\
&= \sum_{t=1}^\infty \frac{L_{ij}M_j}{\sum_{k\in N} L_{ik}M_k} \Pr(\tau=t \given V_0=i) \\
&= \frac{L_{ij}M_j}{\sum_{k\in N} L_{ik}M_k}.
\qedhere
\end{align*}
\end{proof}

Clearly, the matrix~$P$ is invariant to rescaling~$\A$. The only effect
of rescaling~$\A$ by some constant, say~$h$, is to multiply the mean
$1/\gamma_i$ of the distribution of $\tau_{s+1}-\tau_s \given V_{\tau_s}=i$
by a factor of $1/h$. Of course, while such scaling  preserves the persymmetry
of~$\A$, it only makes sense if $0<h< \min\{1/\gamma_i \suchthat i\in N\}$.

\subsection{Alternative interpretation}

There is another way to represent how new nucleotides are added
to a DNA sequence which provides an alternative derivation of the Markov chain
on~$N$ with one-step transition matrix~$P$ of the form \eqref{p.form}.
Let $(Y_s,\ s\geq0)$ be a Markov chain on the set of nucleotides~$N$ with transition matrix $K=(K_{ij})_{i,j\in N}$ given by
$K_{ij}=L_{ij}\big/\sum_{k\in N}L_{ik}$. Thus, the one-step transition matrix of
$(Y_s,\ s\geq0)$ is obtained by converting the positive persymmetric~$\A$ into a
stochastic matrix by normalizing its
rows to sum to unity.
Next, let $(B_s,\ s\geq0)$ be a Bernoulli scheme on~$N$ with common distribution
$M$.
Since $(Y_s,\ s\geq0)$ is a positive recurrent Markov chain on the finite state
space~$N$ and $(B_s,\ s\geq0)$ is an i.i.d. sequence also on~$N$ that is
independent of $(Y_s,\ s\geq0)$, the joint process $\left(\bigl(Y_s,B_s\bigr),\
s\geq0\right)$ is a positive recurrent Markov chain on the state space $N\times
N$ with one-step transition matrix $\left(R_{(i,k),(j,l)}
\right)_{(i,k), (j,l)\in N^2}$ given by $R_{(i,k),(j,l)} = K_{ij}M_l$.

We shall assume without loss of generality that $Y_0=B_0$. Define a sequence of
stopping times $(T_s,\ s\geq0)$ by $T_0:=0$ and
$$
T_{s+1}:=\min\{t>T_s+1 \suchthat Y_{t-1}=B_{t-1}=Y_{T_s} \text{ and }
Y_t=B_t\},
$$
for $s\geq0$. By definition, $Y_{T_s}=B_{T_s}$ for all $s\geq0$ and
$Y_{T_s-1}=B_{T_s-1}$ for all $s\geq1$.
Observe that if $Y_{T_s}$ and $B_{T_s}$ are given, for example, $Y_{T_s}=B_{T_s}=i$, then
\begin{align*}
T_{s+1}-T_s
&=\min\{t>T_s+1 \suchthat Y_{t-1}=B_{t-1}=Y_{T_s} \text{ and }
Y_t=B_t\} - T_s
\\
&=\min\{t>1 \suchthat Y_{t-1}=B_{t-1}=i \text{ and } Y_t=B_t\}.
\end{align*}
    Thus, $T_{s+1}-T_s$ is independent of $T_s$ if
$Y_{T_s}$ is given.
Furthermore, $T_{s+1}-T_s \given Y_{T_s}=i$ has the same distribution as
$T_1\given Y_0=i$. Thus, each $T_s$ is a renewal time at which the Markov chain
$\left(\bigl(Y_s,B_s\bigr),\ s\geq0\right)$ is restarted.

Next, define the stochastic process $(X_s,\ s\geq0)$ by $X_s:=Y_{T_s}$. Since
$\left( \bigl(Y_s,B_s\bigr),\ s\geq0\right)$ is a Markov chain and $(T_s,\
s\geq0)$ is a sequence of stopping times at which it renews, one may employ the
strong Markov property to deduce that $(X_s,\ s\geq0)$ is also a Markov chain.
It only remains to compute its one-step transition matrix.

\begin{theo}
\label{thm:interp2}
The Markov chain $(X_s,\ s\geq0)$ has
one-step transition matrix $P=\left(P_{ij}\right)_{i,j\in N}$, where
$$
P_{ij}:=\frac{L_{ij}M_j}{\sum_{k\in N} L_{ik}M_k}.
$$
\end{theo}

\begin{proof}
Fix $X_0=B_0=i$ and let $T:=T_1$. Then,
\begin{align*}
P_{ij} &= \Pr(X_1=j\given X_0=i) \\
&= \sum_{t=2}^\infty \Pr(Y_T=j, T=t \given X_0=i) \\
  &= \sum_{t=2}^\infty \Pr(Y_t=j \given T=t, X_0=i)\Pr(T=t \given X_0=i) \\
&= \sum_{t=2}^\infty \Pr(Y_t=j, B_t=j \given Y_t=B_t, Y_{t-1}=i, B_{t-1}=i,
Y_{t-2}\neq B_{t-2}, \ldots, Y_2\neq B_2, Y_1\neq B_1, Y_0=i, B_0=i)
\cdot \\
& \quad \Pr(T=t\given X_0=i) \\
&= \sum_{t=2}^\infty \Pr(Y_t=j, B_t=j \given Y_t=B_t,
Y_{t-1}=i)\Pr(T=t\given X_0=i)
\\
&= \sum_{t=2}^\infty \frac{\Pr(Y_t=j, B_t=j \given
Y_{t-1}=i)}{\Pr(Y_t=B_t \given Y_{t-1}=i)}\Pr(T=t\given
X_0=i)
\\
&= \sum_{t=2}^\infty \frac{\Pr(Y_t=j, B_t=j
\given Y_{t-1}=i)}{\sum_{k\in N}\Pr(Y_t=k, B_t=k \given
Y_{t-1}=i)}\Pr(T=t\given X_0=i)
\\
&= \sum_{t=2}^\infty \frac{K_{ij}M_j}{\sum_{k\in N} K_{ik}M_k}\Pr(T=t\given
X_0=i) \\
&= \frac{L_{ij}M_j}{\sum_{k\in N} L_{ik}M_k} \sum_{t=2}^\infty \Pr(T=t\given
X_0=i) \\
&= \frac{L_{ij}M_j}{\sum_{k\in N} L_{ik}M_k},
\end{align*}
since
$$
\sum_{t=2}^\infty \Pr(T=t\given X_0=i)=1
$$
and
\[
\Pr\{Y_t=j, B_t=j \given Y_{t-1}=i)
=\Pr\{Y_t=j \given Y_{t-1}=i)\Pr(B_t=j) =K_{ij}M_j.
\qedhere
\]
\end{proof}

Thus, the mechanism by which nucleotides are appended to a DNA sequence
according to a Markov chain with transition matrix~$P$ may also be described as
follows.
Suppose that the last nucleotide in the sequence is of type~$i$. Then, one
simply waits until both the Markov chain $(Y_s)$ and the i.i.d. sequence $(B_s)$
simultaneously return to state~$i$ and both immediately jump to the same state,
say~$j$. When such a consecutive pair of concordant events occurs, a nucleotide
of type~$j$ is appended to the sequence. At this point, this scheme is repeated,
but using~$j$ as the initial state, so that one waits for A coincident return of
the two processes to state~$j$ followed by simultaneous transitions to a new
state, say~$k$, and so on. The Markov chain $Y_s$ transitions from~$i$ to~$j$
with probability $K_{ij}$ while $B_s$ selects~$j$ with probability~$M_j$.
In contrast to the original description given in~\cite{sobottka&hart2011} and in
Section~\ref{sec:intro}, two nucleotides of types~$j$ and~$k$ are selected with
probabilities~$M_j$ and~$K_{ij}$ respectively and a nucleotide of type~$j$ is
then appended to the end of the sequence if and only if they are of the same
type.
In essence, the mechanism by which nucleotides are appended to the DNA sequence
can be thought of as carrying out acceptance rejection sampling, by repeatedly
drawing independent sample nucleotides from the distributions $(K_{ij},\ j\in
N)$ and~$M$ until they agree (assuming~$i$ is the type of the nucleotide at the
end of the sequence). In this case, the number of draws needed in order to
obtain a suitable nucleotide is a geometric random variable with mean
$1\big/\sum_{j\in N}K_{ij}M_j$.
The first interpretation also amounts to performing acceptance-rejection
sampling, but with a two-step procedure in which a nucleotide type~$j$ is first
proposed by sampling it from the distribution~$M$ and then is added to the DNA
sequence according to an unfair coin toss with probability~$L_{ij}$.

Finally, we note that if the matrix~$\A$ is rescaled so that
$\sum_{i,j\in N}L_{ij}=1$, it admits the natural interpretation as the
stationary dinucleotide probability distribution, that is,
$$
L_{ij} = \Pr(Y_t=i, Y_{t+1}=j),
\qquad i,j\in N,\ t\geq0.
$$
 As noted above, $\A$ remains persymmetric under this kind of rescaling.

\section{$\aleph$-generated matrices}
\label{sec:aleph.generated}

Let $\sS_4$ be the set of all $4\times4$ stochastic matrices, and let $\aleph$
be the cone of positive persymmetric matrices (matrices
$\A=(L_{ij})_{i,j\in N}$ with positive entries such that
$L_{ij}=L_{\alpha(j)\alpha(i)}$ for all $i,j\in N$). Given $P\in \sS_4$ we will
say that $(P,\pi)$ is a stationary Markov chain if the vector
$\pi=(\pi_i)_{i\in N}$ is such that $\pi P=\pi$.

Let $\digamma:\aleph\times (0,+\infty)\to S_4$ be the map which takes $(\A,m)$
to the matrix $\digamma(\A,m)$, which is  given for all $i,j\in N$ by
\begin{equation}
\label{digamma}
\left(\digamma(\A,m)\right)_{ij} := \frac{L_{ij}M_j}{\sum_{k=1}^4
L_{ik}M_k},
\qquad \text{where}\qquad
M_\ell=\left\{\begin{array}{lll}
m/(2m+2) &\text{, if}& \ell=1,4, \\
1/(2m+2)&\text{, if}& \ell=2,3.
                                                                                                                        \end{array}\right.
\end{equation}

Since $\A$ is a positive matrix and $m>0$, the matrix $\digamma(\A,m)$ is primitive, that is, irreducible and aperiodic.

\begin{defn}
\label{defn:aleph.generated}
We say that $P\in\sS_4$ is $\aleph$-generated if there exist $(\A,m)\in \aleph\times (0,+\infty)$ such that
$P=\digamma(\A,m)$.
\end{defn}

Let $\Phi:\sS_4\times (0,+\infty)\times (0,+\infty)\to\aleph$ be the map defined for all stochastic matrices $P$, $\tilde{m}>0$ and $\tilde{s}>0$ by
\begin{equation}%\small
\label{AlephEstimated}
\Phi(P,\tilde{m},\tilde{s}):=
\tilde{s} \begin{pmatrix}
a^{11}_{_P}\kappa_{_P}/\tilde{m}       & a^{12}_{_P}                   & 1                 & \kappa_{_P}/\tilde{m} \\\\
a^{21}_{_P}                     & a^{22}_{_P}\epsilon_{_P} \tilde{m}   & \epsilon_{_P} \tilde{m}  & 1 \\\\
a^{21}_{_P}a^{42}_{_P}                    & a^{22}_{_P}\epsilon_{_P} a^{32}_{_P} \tilde{m} & a^{22}_{_P}\epsilon_{_P} \tilde{m} & a^{12}_{_P} \\\\
a^{11}_{_P}a^{41}_{_P}\kappa_{_P}/\tilde{m}      & a^{21}_{_P}a^{42}_{_P}                  & a^{21}_{_P}                 & a^{11}_{_P}\kappa_{_P}/\tilde{m}
\end{pmatrix},
\qquad\text{where}\qquad
\begin{array}{lcl}
a^{ij}_{_P}        & := & P_{ij}/P_{i\alpha(j)};\\\\
\kappa_{_P}   & := & P_{14}\big/P_{13};\\\\
\epsilon_{_P} & := & P_{23}\big/P_{24}.
\end{array}
%\normalsize
\end{equation}

From \eqref{digamma} it follows that if $P$ is an $\aleph$-generated matrix for
some $\A=\bigl(L_{ij}\bigr)_{i,j\in N}\in\aleph$ and $m\in(0,+\infty)$,
then the nine ratios that appear in \eqref{AlephEstimated} become:
\begin{equation}\label{nine_ratios}
a^{ij}_{_P}         =   L_{ij}/L_{i\alpha(j)},\qquad   \kappa_{_P}    =
L_{14}m/L_{13} \qquad \text{and}\qquad    \epsilon_{_P}  =   L_{23}/L_{24}m.
\end{equation}

\begin{theo}
\label{digamma_inv}
For any $\aleph$-generated matrix $P$,
$$
\digamma^{-1}(P) = \left\{\bigl(\Phi(P,\tilde{m},\tilde{s}),\tilde{m}\bigr):\ \tilde{m}>0,\ \tilde{s}>0\right\}.
$$
\end{theo}

\begin{proof}
Let $P=\digamma(\A,m)$ for some fixed $\A\in\aleph$ and $m>0$.

Given $\tilde{\A}:=\Phi(P,\tilde{m},\tilde{s})$ for any choice of
$\tilde{m},\tilde{s}>0$, it is straightforward to check that
$\digamma(\tilde{\A},\tilde{m})=\digamma(\A,m)=P$.
Therefore, $\left\{\big(\Phi(P,\tilde{m},\tilde{s}),\tilde{m}\big):\
\tilde{m}>0,\ \tilde{s}>0\right\}\subseteq \digamma^{-1}(P)$.

On the other hand, suppose $\A'=(L'_{ij})_{i,j\in N}\in\aleph$ and $m'>0$ are such that $(\A',m')\in\digamma^{-1}(P)$. Note that, since
$P=\digamma(\A,m)=\digamma(\A',m')$, it follows from \eqref{nine_ratios} that
\begin{equation*}
a^{ij}_{_P}         =  L_{ij}/L_{i\alpha(j)}  =   L'_{ij}/L'_{i\alpha(j)}, \qquad \kappa_{_P}    =  L_{14}m\big/L_{13} =   L'_{14}m'\big/L'_{13}, \qquad \epsilon_{_P}  =  L_{23}\big/L_{24}m =   L'_{23}\big/L'_{24}m'.
\end{equation*}
Hence, $\A'=\Phi(P,m',L'_{13})$ and so
$\digamma^{-1}(P)\subseteq\left\{\bigl(\Phi(P,\tilde{m},\tilde{s}),\tilde{m}\bigr):\
\tilde{m}>0,\ \tilde{s}>0\right\}$,
which completes the proof.\qedhere
\end{proof}

Since $\Phi$ is linear in $\tilde{s}$, instead of working with $\Phi$ we can work with the map $\varphi:\sS_4\times(0,+\infty)\to\aleph$ defined by
\begin{equation}\label{varphi}
\varphi(P,\tilde{m}):=\Phi(P,\tilde{m},1).
\end{equation}
Then, $\digamma^{-1}(P)=\left\{\bigl(\tilde
s\varphi(P,\tilde{m}),\tilde{m}\bigr):\ \tilde{m}>0,\ \tilde{s}>0\right\}$.
The next corollary is a simple consequence of
\eqref{varphi} and Theorem~\ref{digamma_inv}.

\begin{cor}
\label{P-aleph1}
A stochastic matrix $P$ is $\aleph$-generated if and only if $P$ is obtained by
probability-normalizing the rows of the matrix $\A:=\varphi(P,1)$.\qed
\end{cor}

Given a vector $\mathbf{a}=(a_1,a_2,a_3,a_4)\in \R^4$, let $D(\mathbf{a})$ be
the $4\times 4$ diagonal matrix with $\mathbf{a}$ on its diagonal.

\begin{cor}\label{P-aleph2}
A stochastic matrix $P$ is $\aleph$-generated if and only if there exists a
strictly positive vector $\mathbf{x}=(x_i)_{i\in N}\in\R^4$ such that $D(\mathbf{x})P\in\aleph$.
\end{cor}

\begin{proof}
Suppose that $P$ is $\aleph$-generated and define $\mathbf{x}$ to be the vector
with elements given by
$
x_i:=\sum_{k=1}^4\Bigl(\varphi(P,1)\Bigr)_{ik}
$.\sloppy\
Then, from Corollary~\ref{P-aleph1} we have that
$D(\mathbf{x})P=\varphi(P,1)\in\aleph$.

Conversely, if $D(\mathbf{x})P=\A\in\aleph$ for some $\mathbf x\in\R^4$, then
$P=\digamma(\A,1)$ and $\mathbf x$ contains the row sums of $\A$.\qedhere
\end{proof}

Note that given an $\aleph$-generated matrix $P$, there exist infinitely many
vectors $\mathbf{x}$ that satisfy the stated property, all of which
are collinear.
Because of this,
we can decide whether or not a stochastic matrix is $\aleph$-generated
by setting
$$
\mathbf{x}_P = \left(\frac{P_{j\alpha(i)}}{P_{i\alpha(j)}} \right)_{i\in N}
= \frac1{\sum_{k=1}^4\bigl(\varphi(P,1)\bigr)_{jk}} \left(
\sum_{k=1}^4\bigl(\varphi(P,1)\bigr)_{ik} \right)_{i\in N},
$$
for a fixed $j\in\{1,2,3,4\}$, and checking whether or not $D(\mathbf{x}_P)P$
belongs to $\aleph$. Observe that $\mathbf{x}_P$ is expressed in terms of elements
of~$P$. In particular, we can choose
\begin{equation*}
\mathbf{x}_P=\left(\frac{P_{44}}{P_{11}},\frac{P_{43}}{P_{21}},\frac{P_{42}}{P_{31}},1\right).
\end{equation*}

\section{$\aleph$-families and generators}
\label{sec:families}

From the preceding discussion, it is evident that a given $\aleph$-generated
stochastic matrix can be generated using any one of a multitude of
persymmetric matrices. We proceed to examine this non-uniqueness in greater
detail.

\begin{defn}
The $\aleph$-family of an $\aleph$-generated matrix $P$ is the set
$$
\aleph(P):=\left\{\varphi(P,\tilde{m}):\ \tilde{m}>0\right\}.
$$

The family of generators of an $\aleph$-generated matrix $P$ is the set
$$
\aleph_G(P):=\left\{\bigl(\varphi(P,\tildem),\tildem\bigr):\ \tildem>0\right\}.
$$
\end{defn}

The import of the next theorem is that $\aleph$ can be partitioned into
equivalence classes. Firstly, any persymmetric matrix can be used
to generate a whole host of $\aleph$-generated matrices simply by varying the
value of the parameter~$m$. Thus, there are families of persymmetric
matrices that give rise to disjoint collections of $\aleph$-generated matrices
and these families are mutually exclusive, partitioning the space~$\aleph$ into
equivalence classes. Secondly, for each $\aleph$-generated matrix~$P$, there is
a set of persymmetric matrices, each of which generates~$P$ when combined with
the appropriate value of~$m$. This leads to an equivalence relation on the set
$\aleph\times (0,\infty)$.

\begin{theo}
Suppose $P$ and $Q$ are two $\aleph$-generated matrices. Then:
\begin{enumerate}
\item Either $\aleph(P)\cap\aleph(Q)=\emptyset$ or $\aleph(P)=\aleph(Q)$.

\item
Either
\begin{enumerate}
\item $\aleph_G(P)\cap\aleph_G(Q)=\emptyset$ and $P\neq Q$; or
\item $\aleph_G(P)=\aleph_G(Q)$ and $P=Q$.
\end{enumerate}
\end{enumerate}
\end{theo}

\begin{proof}
\ \par\nobreak
\begin{enumerate}
\item
Suppose $\aleph(P)\cap\aleph(Q)\neq\emptyset$ and choose an $\A=(L_{ij})_{i,j\in N}\in\aleph(P)\cap\aleph(Q)$. Let $m^{(1)},m^{(2)}\in(0,+\infty)$ be
such that $P=\digamma(\A,m^{(1)})$ and $Q=\digamma(\A,m^{(2)})$.

We begin  by proving that $\aleph(P)\subseteq\aleph(Q)$. Let
$\mathfrak{B}=(B_{ij})_{i,j\in N}\in\aleph(P)$ and let $m^{(3)}>0$ be such that
 $P=\digamma(\mathfrak{B},m^{(3)})$. Since $P$ and $Q$ can be generated by the
 same $\A$, they share the same ratios $a^{ij}_{\cdot}$ listed
 in~\eqref{nine_ratios}, that is,
\begin{equation}
\label{a-h}
a^{ij}_{_P} = a^{ij}_{_Q}
\end{equation}
The other two ratios for $P$ will satisfy the following equalities:
\begin{gather*}
\kappa_{_P}    :=  P_{14}/P_{13} = B_{14}m^{(3)}/B_{13} = L_{14}m^{(1)}/L_{13}, \\
\epsilon_{_P}  :=  P_{23}/P_{24} = B_{23}/B_{24}m^{(3)} = L_{23}/L_{24}m^{(1)},
\end{gather*}
which means that
\begin{equation}
\label{LB}
B_{14}m^{(3)}\big/B_{13}m^{(1)} = L_{14}\big/L_{13},\qquad\text{and}\qquad
B_{23}m^{(1)}\big/B_{24}m^{(3)} = L_{23}\big/L_{24}.
\end{equation}
On the other hand, the last two ratios for $Q$ are:
\begin{equation}
\label{k-e}
\begin{gathered}
\kappa_{_Q}    :=  Q_{14}/Q_{13} =  L_{14}m^{(2)}/L_{13} = B_{14}m^{(3)}m^{(2)}/B_{13}m^{(1)} = \frac{m^{(2)}}{m^{(1)}}\kappa_{_P} ,\\
\epsilon_{_Q}  :=  Q_{23}/Q_{24} =  L_{23}/L_{24}m^{(2)} = B_{23}m^{(1)}/B_{24}m^{(3)}m^{(2)} = \frac{m^{(1)}}{m^{(2)}}\epsilon_{_P},
\end{gathered}
\end{equation}
where the last equality in each line follows from \eqref{LB}.

Setting $\tildem:= m^{(2)}m^{(3)}/m^{(1)}$, and taking \eqref{a-h} and
\eqref{k-e} together with the last line in the proof of Theorem
\ref{digamma_inv} yields  $\varphi(Q,\tildem) =\varphi(P, m^{(3)}) =
\mathfrak{B}$. Therefore, $\mathfrak{B}\in\aleph(Q)$ and so
$\aleph(P)\subseteq\aleph(Q)$.

 Next, let $\mathfrak B\in\aleph(Q)$. By symmetry, another application
 of the above argument allows us to conclude that $B\in\aleph(P)$ and
 hence $\aleph(Q)\subseteq\aleph(P)$. Therefore, $\aleph(P)=\aleph(Q)$.

\item
By definition, either $\digamma^{-1}(P)=\digamma^{-1}(Q)$, in which case $P=Q$,
or $\digamma^{-1}(P)\cap\digamma^{-1}(Q)=\emptyset$ and $P\neq Q$. Now,
$\aleph_G(P) \subset \digamma^{-1}(P)$ since $\digamma^{-1}(P) = \left\{
\tildes\A \suchthat \tildes>0,\ \A\in\aleph_G(P) \right\}$, and the result
follows.
\qedhere
\end{enumerate}
\end{proof}

\begin{defn}
Given an $\aleph$-generated matrix $P$ and $\tilde{m}\in (0,+\infty)$, we define
the {\em $\tilde{m}$-canonical representative} of $\aleph(P)$ to be the matrix $\A_{P,\tilde{m}}:=\varphi(P,\tilde{m}/\epsilon_{_P})$.
\end{defn}

Note that $\bigl(\A_{P,\tilde{m}}, \tilde{m}/\epsilon_{_P}\bigr)$ is a generator
of $P$.
Furthermore, from \eqref{nine_ratios} if $P$ and $Q$ are two $\aleph$-generated
matrices with $\aleph(P)=\aleph(Q)$, then $a^{ij}_{_P}=a^{ij}_{_Q}$, for all
$i,j$ and $\kappa_{_P}\epsilon_{_P}=\kappa_{_Q}\epsilon_{_Q}$.
This gives

\begin{cor}\label{cor canonical_representative} Two
$\aleph$-generated matrices belong to the same $\aleph$-family if and
only if they have identical canonical representatives, that is, if
$P$ and $Q$ are $\aleph$-generated, then
$$
\aleph(P)=\aleph(Q)\qquad\Longleftrightarrow\qquad \A_{P,1}=\A_{Q,1}\quad
\Longleftrightarrow \quad \A_{P,\tilde{m}}=\A_{Q,\tilde{m}}, \text{ for all }
\tilde{m}>0. \qed
$$
\end{cor}

\section{Properties of $\aleph$-generated matrices}
\label{sec:properties}

Given the stationary Markov chain $(P,\pi)$, consider the following related
stationary Markov chains: $(P^\alpha,\pi^\alpha)$ is the complement
Markov chain of $(P,\pi)$, where $P_{ij}^\alpha := P_{\alpha(i)\alpha(j)}$ and
$\pi^\alpha_i:=\pi_{\alpha(i)}$; $(P^*,\pi^*)$ denotes the reverse Markov chain of
$(P,\pi)$, where $P_{ij}^*:=\pi_jP_{ji}\big/\pi_i$ and $\pi^*_i:=\pi_{i}$; and
$(\tilde P,\tilde \pi)$ is the reverse complement Markov chain of $(P,\pi)$, where
$\tilde P_{ij} = \pi_{\alpha(j)}P_{\alpha(j)\alpha(i)}\big/\pi_{\alpha(i)}$ and
$\tilde\pi_i=\pi_{\alpha(i)}$. Note that $\tilde P=(P^\alpha)^*= (P^*)^\alpha$
and   $\tilde \pi=(\pi^\alpha)^*=(\pi^*)^\alpha$. The names complement, reverse
and reverse complement come from the genetics and Markov chains literature,
referring to several kinds of  relationship that can exist between nucleotide
sequences (genetics), as well as Markov chains.

\begin{theo}\label{all_or_none}
The matrices~$P$, $P^\alpha$, $P^*$ and $\tilde P$ are either all
$\aleph$-generated or none of them are.
\end{theo}

\begin{proof}
Assume $P$ is $\aleph$-generated and take
$\A:=\varphi(P,1)=\bigl(L_{ij}\bigr)_{i,j\in N}$.

Define $\A^\alpha=(L^\alpha_{ij})_{i,j\in N} \in \aleph$, where
$L_{ij}^\alpha:=L_{\alpha(i)\alpha(j)}$. Then,
\begin{align*}
P_{ij}^\alpha &= P_{\alpha(i)\alpha(j)}
= \frac{L_{\alpha(i)\alpha(j)}}{\sum_{k=1}^4L_{\alpha(i)k}}
= \frac{L_{ij}^\alpha}{\sum_{k=1}^4L_{ik}^\alpha}, \quad i,j\in N
\end{align*}
 and  $P^\alpha$ is $\aleph$-generated
with $P^\alpha=\digamma\bigl(\A^\alpha, 1\bigr)$.

To check that $P$ is $\aleph$-generated implies that $P^*$ is
also $\aleph$-generated, it suffices  by Corollary \ref{P-aleph2} to set
$\ds
\mathbf{x}_{P^*}=\left(\frac{P_{44}^*}{P_{11}^*},\frac{P_{43}^*}{P_{21}^*},\frac{P_{42}^*}{P_{31}^*},1\right)$
and prove that $D\bigl(\mathbf{x}_{P^*})P^*\in\aleph$. In fact,
$\left(D\bigl(\mathbf{x}_{P^*}\bigr)P^*\right)_{ij}=\left(D\bigl(\mathbf{x}_{P^*}\bigr)P^*\right)_{\alpha(j)\alpha(i)}$
because
\begin{align*}
\left(D\bigl(\mathbf{x}_{P^*}\bigr)P^*\right)_{ij}
&= (\mathbf{x}_{P^*})_i P_{ij}^* = \frac{P_{4\alpha(i)}^*}{P_{i1}^*}P_{ij}^*
= \frac{\frac{\pi_{\alpha(i)}}{\pi_4}\  P_{\alpha(i)4}}{\frac{\pi_1}{\pi_i}\  P_{1i}}\frac{\pi_j}{\pi_i}\  P_{ji}
=\frac{\pi_{\alpha(i)}\pi_j}{\pi_1\pi_4}\frac{P_{\alpha(i)4}}{P_{1i}}
P_{ji}
\end{align*}
and similarly
$$
\left(D\bigl(\mathbf{x}_{P^*}\bigr)P^*\right)_{\alpha(j)\alpha(i)}
= \frac{\pi_{j}\pi_{\alpha(i)}}{\pi_1\pi_4}\frac{P_{j4}}{P_{1\alpha(j)}}
P_{\alpha(i)\alpha(j)},
$$
while
\begin{align*}%\small
\frac{P_{\alpha(i)4}}{P_{1i}}P_{ji}
&=\frac{L_{\alpha(i)4}}{\sum_{k=1}^4 L_{\alpha(i)k}}
\frac{L_{ji}}{L_{1i}} \frac{\sum_{k=1}^4 L_{1k}}{\sum_{k=1}^4 L_{jk}}
=\frac{L_{\alpha(i)4}}{\sum_{k=1}^4 L_{\alpha(i)k}}
\frac{L_{\alpha(i)\alpha(j)}}{L_{\alpha(i)4}} \frac{\sum_{k=1}^4 L_{1k}}{\sum_{k=1}^4 L_{jk}} \\\\
&=\frac{L_{\alpha(i)\alpha(j)}}{\sum_{k=1}^4 L_{\alpha(i)k}}
\frac{\sum_{k=1}^4 L_{1k}}{\sum_{k=1}^4 L_{jk}}
=\frac{L_{\alpha(i)\alpha(j)}}{\sum_{k=1}^4 L_{\alpha(i)k}}
\frac{L_{j4}}{L_{1\alpha(j)}} \frac{\sum_{k=1}^4 L_{1k}}{\sum_{k=1}^4 L_{jk}}
=\frac{P_{j4}}{P_{1\alpha(j)}} P_{\alpha(i)\alpha(j)}.
\end{align*}%\normalsize

Next, to check that $\tilde P$ is $\aleph$-generated given that $P$  is
$\aleph$-generated, we need only note that $\tilde P=(P^\alpha)^*$ and
apply the above two results one after the other.

The proof is completed by realising that
$P=(P^\alpha)^\alpha=(P^*)^*=\tilde{\tilde P}$ and hence being
$\aleph$-generated is a solidarity property of the four matrices.
\qedhere
\end{proof}

Most bacterial DNA sequences can be segmented into two halves called
chirochores~\cite{FrankLobry00} and the two stationary Markov chains that
empirically approximate their first-order structure are
reverse complements of each other \cite{sobottka&hart2011}. If the DNA
sequence conforms to the S-H model then the dinucleotide distribution in one of
the chirochores is approximated by $(P,\pi)$ with $P$ being $\aleph$-generated.
However, it was an open question as to whether or not the other chirochore would
also be approximated by an $\aleph$-generated Markov chain. Theorem \ref{all_or_none} above answers
this question in the positive.

Furthermore, it is common to find that the stationary Markov chain
$(W,\omega)$ that approximates the first-order structure of an entire DNA
sequence satisfies {\em intra-strand parity}
\cite{AlbrechtBuehler2006,hart&martinez2011}, that is,
$\omega_iW_{ij}=\omega_{\alpha(j)}W_{\alpha(j)\alpha(i)}=\tilde\omega_i\tilde
W_{ij}$ for all $i,j\in N$.
Intra-strand parity has been observed in the DNA
sequences of many organisms such as Bacteria, archaea, plants and animals, but
not in other sequences such as those from single-stranded viruses and
organelles.
The next theorem relates intra-strand parity of dinucleotides to
the $\aleph$-generated matrices (cf. the direct characterization
in~\cite[Proposition 1]{hart&martinez2011}) and shows that $\aleph$-generated matrices satisfy a weaker property than intra-strand
parity.

\begin{theo}\label{ISP<->aleph-generated}
Let $(W,\omega)$ be a stationary Markov chain. Then $(W,\omega)$ satisfies
$\omega_iW_{ij}=\omega_{\alpha(j)}W_{\alpha(j)\alpha(i)}$ for all $i,j\in N$ if and only if it is $\aleph$-generated and
the matrix $\A:=\varphi(W,1)=\bigl(L_{ij}\bigr)_{i,j\in N}$ that
generates it satisfies $S_i=S_{\alpha(i)}$ for $i\in N$, where
$S_i:=\sum_{k=1}^4L_{ik}$.
Furthermore, if $W$ complies with intra-strand parity, then its stationary
distribution~$\omega$ can be explicitly expressed as
$\omega=\frac1{2(S_1+S_2)}(S_1,S_2,S_2, S_1)$.
\end{theo}

\begin{proof}
\ \par\nobreak
\noindent [$(\Longrightarrow)$]
It can be seen that $W$ is $\aleph$-generated by observing that
$\bigl(D(\omega)W\bigr)_{ij}=\omega_iW_{ij}=\omega_{\alpha(j)}W_{\alpha(j)\alpha(i)}=\bigl(D(\omega)W\bigr)_{\alpha(j)\alpha(i)}$.
Next, let $\A=\varphi(W,1)$. One can easily check that
$\omega_iW_{ij}=\omega_{\alpha(j)}W_{\alpha(j)\alpha(i)}$, for $i,j\in N$, implies $\omega_i=\omega_{\alpha(i)}$ for all $i\in N$. Therefore,
$
\omega_i\frac{L_{ii}}{\sum_{k=1}^4L_{ik}}
= \omega_{\alpha(i)}\frac{L_{\alpha(i)\alpha(i)}}{\sum_{k=1}^4L_{\alpha(i) k}}
= \omega_i\frac{L_{ii}}{\sum_{k=1}^4L_{\alpha(i) k}},
$
for all $i\in N$, and hence $\sum_{k=1}^4L_{ik}=\sum_{k=1}^4L_{\alpha(i)k}$.\sloppy

\noindent [$(\Longleftarrow)$]
Suppose $W$ is obtained by normalizing the rows of a
matrix $\A=(L_{ij})_{i,j\in N}\in\aleph$, that is, $
W=\left(\frac{L_{ij}}{S_i}\right)_{i,j\in N}
$,
 where $S_i:=\sum_{k=1}^4L_{ik}$. Suppose that $\A$
satisfies $S_i=S_{\alpha(i)}$ for $i=1,2$.
It is easy to check that
$\omega:=\frac1{2(S_1+S_2)} (S_1,S_2,S_2,S_1)$
is the stationary distribution of~$W$. Hence, it follows that for all $i,j\in
N$,
\begin{equation*}
\omega_iW_{ij}=\frac{S_i}{2(S_1+S_2)}\frac{L_{ij}}{S_i}
=\frac{L_{ij}}{2(S_1+S_2)}
=\frac{S_{\alpha(j)}}{2(S_1+S_2)}\frac{L_{\alpha(j)\alpha(i)}}{S_{\alpha(j)}}
=\omega_{\alpha(j)}W_{\alpha(j)\alpha(i)}.
\qedhere
\end{equation*}
\end{proof}

\section{Applications and final remarks}
\label{sec:conclusion}

This article has given a mathematical analysis
of the S-H model and elucidated its properties. We conclude with some
remarks about the application of the results that have been presented here. Corollary~\ref{cor canonical_representative} provides a way of deciding whether or not  two or more $\aleph$-generated
matrices can be generated from a single persymmetric matrix $\A$ in conjunction
with different values of the parameter~$m$. Meanwhile,
Theorem \ref{ISP<->aleph-generated}
shows that intra-strand parity in dinucleotides is a special case of
$\aleph$-generated matrices. Possessing a weaker structure than that
encapsulated by intra-strand parity, it is possible that $\aleph$-generated
matrices may be useful for capturing the dinucleotide structure in genomic
sequences that do not exhibit intra-strand parity.

For the purposes of applications, corollaries~\ref{P-aleph1}
and~\ref{P-aleph2} are useful for constructing measures of how close the
estimated stationary Markov chain of a bacterial DNA sequence is to being
$\aleph$-generated. Given $P\in \sS_4$, we can define the following two
examples of such measures:

\noindent{\bf Measure 1:} Let $\proj(Q)$ be the orthogonal projection of a
$4\times 4$ positive matrix~$Q$ onto~$\aleph$, and define
$$\delta_1(P):=\min_{\mathbf{x}=(x_1,x_2,x_3,1)}
\norm{D(\mathbf{x})P-\proj\Bigl(D(\mathbf{x})P\Bigr)}.
$$
The quantity $\delta_1(P)=0$ if and only if $P$ is $\aleph$-generated.
Otherwise, $\delta_1(P)$ gives the minimal distance between some matrix
$D(\mathbf{x})P$ which generates $P$ according to  the model (but which does not
belong to~$\aleph$) and the space~$\aleph$. Note that $\delta_1(P)$ can be
analytically computed. The minimum in the expression for $\delta_1(P)$ is
attained at the point $\mathbf{x}=(x_1,x_2,x_3,1)$, where
$$
\begin{pmatrix}
x_1   \\
x_2   \\
x_3
 \end{pmatrix}
 =
\begin{pmatrix}
p_{11}^2+p_{12}^2+p_{13}^2  & -p_{24}p_{13}                & -p_{34}p_{12}                \\ -p_{13}p_{24}                &  p_{21}^2+p_{22}^2+p_{24}^2  & -p_{33}p_{22}                \\
-p_{12}p_{34}                & -p_{22}p_{33}                &  p_{31}^2+p_{33}^2+p_{34}^2
\end{pmatrix}^{-1}
 \begin{pmatrix}
  p_{44}p_{11}   \\
 p_{43}p_{21}   \\
 p_{42}p_{31}
 \end{pmatrix}.
$$

\noindent{\bf Measure 2:} Let $\epsilon$ be a $4\times 4$ matrix, $P(\epsilon):=P+\epsilon$, and
$\mathbf{x}=(x_1,x_2,x_3,1)$ be a positive vector. Define $\delta_2(P)$  as
the solution of the following optimization problem:
$$
\text{min } \sum_{i,j\in N} \epsilon_{i,j}^2 \qquad\text{subject to
}
\left\{\begin{array}{l}
P(\epsilon) \in\sS_4;\\
D(\mathbf{x})P(\epsilon)-\proj\Bigl(D(\mathbf{x})P(\epsilon)\Bigr)=\mathbf{0}.
\end{array}\right.
$$
As was the case with $\delta_1(P)$, we have $\delta_2(P)=0$ if and only if
$P$ is $\aleph$-generated, otherwise, $\delta_2(P)$ gives the shortest
squared Frobenius distance between $P$ and some $\aleph$-generated stochastic
matrix.
There being no closed-form solution to the optimization problem, the computation of
$\delta_2(P)$ would need to be implemented using numerical methods.

Finally, the development of statistical hypothesis tests based on these
measures together with  further statistical analyses and their application to real
bacterial genomes are planned for future publication.

\section*{Acknowledgments}
This work was supported by the Center for Mathematical Modeling  CONICYT
Project/Grant PIA AFB 170001, Fondecyt Regular Grant 1070344 and
CNPq-Brazil grants 308575/2015-6 and 301445/2018-4. M. Sobotka was
partially supported by CNPq-Brazil grant
54091/2017-6. Part of this work was carried out while M. Sobotka was
visiting the Center for Mathematical Modeling at the University of Chile.

\small

\bibliographystyle{plain}
%\bibliography{maths,HMCiBDNA}

\end{document}